\documentclass[conference]{IEEEtran}
\usepackage{cite}
\usepackage{amsmath,amssymb,amsfonts}
\usepackage{amsthm}
\usepackage{graphicx}
\usepackage{epstopdf}
\usepackage{textcomp}
\usepackage{xcolor}

\usepackage{tikz}

\newtheorem{theorem}{Theorem}

\newtheorem{corollary}{Corollary}

\usepackage{algorithm, algpseudocode}

\def\BibTeX{{\rm B\kern-.05em{\sc i\kern-.025em b}\kern-.08em
    T\kern-.1667em\lower.7ex\hbox{E}\kern-.125emX}}
\begin{document}

\title{Choosing optimal parameters for a distributed multi-constrained QoS routing}

\author{\IEEEauthorblockN{1\textsuperscript{st} Sergey Komech}
\IEEEauthorblockA{\textit{Institute for Information Transmission Problems} \\
Moscow, Russia \\
komech@iitp.ru \\
ORCID 0000-0002-8261-510X}
\and
\IEEEauthorblockN{2\textsuperscript{nd} Andrey Kupavskii}
\IEEEauthorblockA{\textit{Moscow Institute of Physics and Technology} \\
Dolgoprudny, Russia \\
kupavskii@ya.ru \\
ORCID 0000-0002-8313-9598}
\and
\IEEEauthorblockN{3\textsuperscript{rd} Alexei V Vezolainen}
\IEEEauthorblockA{\textit{Huawei Russian Research Institute} \\
Moscow, Russia \\
alexei.v.vezolainen@huawei.com \\
ORCID 0000-0002-1498-3920}
}

\maketitle

\begin{abstract}
We consider several basic questions on distributed routing in directed graphs with multiple additive costs,
or metrics, and multiple constraints.
Distributed routing in this sense is used in several protocols, such as IS-IS and OSPF.
A practical approach to the multi-constraint routing problem is to, first,
combine the metrics into a single `composite' metric, and then apply
one-to-all shortest path algorithms, e.g. Dijkstra, in order to find shortest path trees.
We show that, in general, even if a feasible path exists and is known for every source and destination pair,
it is impossible to guarantee a distributed routing under several constraints.
We also study the question of choosing the optimal `composite' metric.
We show that under certain mathematical assumptions we can efficiently find a convex
combination of several metrics that maximizes the number of discovered feasible paths.
Sometimes it can be done analytically, and is in general possible using what we call a 'smart iterative approach'.
We illustrate these findings by extensive experiments on several typical network topologies.
\end{abstract}



\section{Introduction}

A fundamental problem of network routing is to find a path which satisfies multiple constraints such as delay, packet loss rate, cost, bandwidth etc.
The multi-constrained routing problem is well-known and widely investigated
using different approaches such as Multi-Constrained Path,
Multi-Constrained Shortest Path, or Multi-Constrained Optimal Path problem
\cite{2003Kuipers}, \cite{2011Multi,LOZANO2013}, \cite{2000Tamcra,2001Optimal,2015MCMSP}.


We study this question from the perspective of distributed routing, which means that at each node some deterministic routing table in the format (destination, next hop) is used to forward packets. (At the same time, we make no assumption on the algorithm that has generated such a table.) We show that, in general, there is no distributed routing solution that could provide us with, say, feasible paths for more than 50\% of source-destination pairs, even
 if feasible paths exist for every such pair. We provide formal proofs along with several examples.

On the positive side, we study the choice of an optimal  weight $p$ for combining two weight functions into a new one.
We show that under some mathematical assumptions the corresponding optimal $p$ exists.
We study single the resulting composite metric as a function of $p$ and show its convexity in a general case with multiple constraints.  Our results generalize the $2$-constraints case studied in \cite{Puri}.

Comparing it with previous studies, while aforementioned approaches aim at finding feasible paths between source and destination, our approach is designed to find an optimal parameter $p$ which then can be used by usual Shortest Path algorithms.

Two different schemes of routing are proposed. The first one relies on a unique $p$-value for all the routers. The second one allows routers to have a set of  ``optimal'' $p$'s. Both approaches seem to be applicable to distributed routing with multiple constraints.

We propose a simple estimate for an optimal parameter $p$ for uniformly distributed random costs. It could also be used as a rule of thumb for a quick choice
of $p$ for normal distributions of costs. For a general case, we describe an iterative approach to find an optimal value of parameter $p$ for a distributed routing.

\section{Related Work}

Providing near-optimal routing with small tables and labels is a fundamental challenge in both centralized and distributed routing modes.
Recently, in \cite{2020multiplerouting} the authors studied limitations of standard vectoring protocols such as EIGRP, BGP, DSDV, and Babel, see \cite{EIGRP},\cite{bgp2006},\cite{dsdv}, \cite{babel} respectively. From a mathematical point of view, limitations occur since standard vectoring protocols fail to find optimal paths if extension operation is not isotone for the total order on a set of paths, e.g.  \cite{2003metrics, 2019nonisotonic}. The authors in \cite{2020multiplerouting}  designed a new vectoring protocol based on partial order that satisfied isotonicity and respected  total order. This approach  leads to an increase in the size of routing tables, because a single path is replaced by a set of incomparable paths. Intensive research on balancing between near-optimality and optimizing the size of routing tables lasts for years.

A well-known approach to multiple constraints routing is to consider a linear combination of costs as a new single composite weight function and to find the corresponding shortest paths \cite{Jaffe,Puri,Multi-routing-2constraint-single2004} or $k$-shortest paths, as in \cite{2000Tamcra,2015MCMSP}. Authors considered an iterative algorithm for finding paths using a shortest path algorithm with a linear combination of weights as a new weight function \cite{Cui}, \cite{cui2003precomputation}, \cite{Curado}.

In \cite{Puri} the problem of routing under multiple constraints was considered and some algorithms were presented for
a graph where each edge is labeled with a cost and a delay.
The authors showed that the problem is NP-hard and presented a pseudo-polynomial time
algorithm which solved the problem exactly.
The algorithm finds a path satisfying the constraints or states that there
is no such path.
The authors developed an algorithm for finding a path which satisfies the two constraints using the shortest path algorithm for the weight $\alpha\cdot$cost$+(1-\alpha)\cdot$delay. Unlike the former, the latter algorithm does not guarantee finding a feasible path, even if one exists.

In \cite{2016Wang} the authors proposed a Multi-Constrained Multi-Path Problem
to find an optimal path and sub-optimal paths for routing under multi-constraints for some particular
Data Center Networks.

\section{Challenges in distributed multi-constrained routing}
In this section, we show that distributed routing poses additional challenges in multi-constrained setting. Formally, assume that each edge in the network has two costs, and the goal is to find certain paths that satisfy a constraint in each of the two costs. We will show that, even if the required paths exist and are easy to find in a network and, moreover, have certain slack in both constraints, in general it  is impossible to provide a distributed solution that will connect a majority of pairs of vertices with feasible paths. We will provide sharp bounds in terms of the ratio of feasible paths, as well as the slack in the costs.

We assume that the costs are independent and additive, that is, the cost on a path is equal to the sum of costs on its edges.
Recall that we work with distributed routing, i.e. we assume that each vertex in the network has a routing table of the form (destination, next hop/next edge)
and routes packets using only this table. This is the only assumption on the algorithms.
We call such class of algorithms {\it Distributed Routing Algorithms}.

Consider the following network $G_2 = (V,E)$ with vertices  $V = \{c,v_1,\ldots, v_{2n}\}$, edges $E$, and two costs $(w_1, w_2)$ on each edge.
Let $E = \cup_{i=1}^{2n} E_i$, where $E_i$ consists of three edges connecting $c$ and $v_i$: two edges $cv_i$ with costs $(1,0)$ and $(0,1)$, respectively, and one edge $v_ic$ with cost $(1,0)$ if $i\le n$ and cost $(0,1)$ if $i>n$.
The corresponding oriented multigraph is shown in the Figure~\ref{not100-2}.

\begin{figure}[!h]
\centering
\begin{tikzpicture}[> = stealth, 
            shorten > = 1pt
                   ]

  \node[circle,draw=red,thick] (O) at (0,0){$c$} ;
  \node[circle,draw] (A) at (3,0) {$v_i$} ;
  \node[circle,draw] (B) at (-3,0) {$v_{n+i}$};




  \draw (O) edge[bend left=50,->]  node [above] {$(0;1)$}  (A);
  \draw (O) edge[bend right=50,->] node [below] {$(1;0)$} (A);
  \draw (A) edge[->] node [below] {$(1;0)$} (O);

   \draw (O) edge[bend left=50,->]  node [below] {$(1;0)$}  (B);
  \draw (O) edge[bend right=50,->] node [above] {$(0;1)$} (B);
  \draw (B) edge[->] node [above] {$(0;1)$} (O);

\end{tikzpicture}

\caption{No guarantee for the distributed multi-constrained routing is possible for such a scenario. The vertices $v_i$, $1\le i\le n$,  have an edge toward $c$ with costs $(w_1;w_2):=(1;0)$; the vertices  $v_{n+i}$, $1\le i\le n$, have an edge towards $c$ with costs $(0;1)$.}
\label{not100-2}
\end{figure}
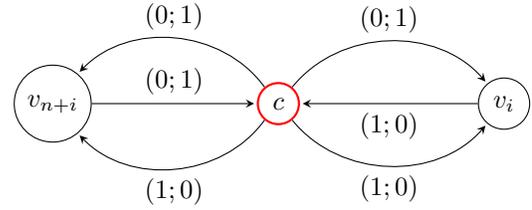


\begin{theorem}\label{th1}
 Assume that $A$ is a distributed routing algorithm that aims to find paths between any two pairs of vertices in $G_2$ that satisfy requirements $(w_1,w_2)< (2,2)$, where $w_i$ is the sum of the $i$-th coordinates of cost vectors along the path. (We call such paths {\emph satisfying}.) Then

 \begin{itemize}
   \item[(i)] There is a satisfying path between any two vertices;
   \item[(ii)] the algorithm $A$ cannot find satisfying paths for more than $(2n+1)^2-2n(n-1)$ ordered pairs of vertices out of $(2n+1)^2$. That is, asymptotically we can find satisfying paths for no more than $0.5+o(1)$ fraction of pairs.
 \end{itemize}
\end{theorem}
\begin{proof}
The first part should be obvious: there are paths of costs at most $(1,1)$ for any pair of vertices.
We have to show that for any $A$ the number of pairs of vertices with no satisfying path found by $A$ is at least $2n(n-1)$. Indeed,  for each $v_j$ as the destination the vertex $c$ has to choose either the edge with cost $(1,0)$ or the edge with cost $(0,1)$ as the next edge. In the former case, $v_j$ is unreachable for the vertices $v_i$ with $i\le n$, possibly excluding $v_j$ itself. Similarly, in the second case $v_j$ is unreachable for the vertices $v_i$ with $i>n$, possibly excluding $v_j$ itself. One way or another, there are at least $n-1$ vertices from which $v_j$ is unreachable via paths provided by $A$. Summing over all $v_j$, we get the result.
\end{proof}
Note that in the previous theorem  there were paths with costs almost twice smaller than the constraint in both coordinates. The next theorem shows that this is tight.
\begin{theorem}\label{th2}
 Let $G = ([n], E)$ be a network with two non-negative costs at each edge. Fix a QoS restriction for each pair of nodes $S =(u,v)$
 in the form $\le (2a_P,2a_P)$, which means that we need to find a $(u,v)$-path of the cost at most $2a_P$ in both costs. Let $\mathcal P\subset {[n]\choose 2}$ be a collection of pairs $P=(v,w)$ such that there is a $(u,v)$-path of cost at most $(a_P, a_P)$.
Then there is a \textbf{deterministic} distributed algorithm $A$ such that
for each $P=(u,v)\in \mathcal P $ we will be able to find a feasible $(u,v)$-path using $A$.
\end{theorem}
\begin{proof} We define $A$ as follows. First, we define a single weight on each edge $e=(x,y)\in E$ with weights $(w^1_e,w_e^2)$ as follows: $w_e = w_e^1+w_e^2$. Then we run a standard one-to-all shortest path algorithm (e.g., Dijkstra) on the graph and route according to the obtained shortest path trees.

We claim that the resulting path satisfies the original constraint. Indeed, assume that we found a shortest $(u,v)$-path $S$ for $E = (u,v)\in \mathcal P$. The path given by the assumption has weight at most $a_P+a_P = 2a_P.$ Thus the found path has weight at most $2 a_P$, and its weight in each of the constraints is at most $2a_P$, which means that it is feasible for the original constraints.
\end{proof}

If we look into the case when the constraint is $(w_1,w_2)<(2,2)$ and there is a path of cost at most $(1,1)$, then we show that the bound given in Theorem~\ref{th1} on the proportion of found feasible paths is again best possible.

\begin{theorem}\label{th3}
Let $G = ([n], E)$ be a network with two non-negative costs on each edge. Fix a constraint for each pair $S =(u,v)$ of the form $<(2a_P,2a_P)$, which means that we need to find a $(u,v)$-path of the cost strictly smaller than $2a_P$ in both costs. Let $\mathcal P\subset {[n]\choose 2}$ be a collection of pairs $P=(v,w)$ such that there is a $(u,v)$-path of cost at most $(a_P, a_P)$.
Then there is a \textbf{randomized} distributed algorithm $A$ such that
for at least half of pairs $P=(u,v)\in \mathcal P $ in expectation we will be able to find a feasible $(u,v)$-path using $A$.

As a consequence, there exists a \textbf{deterministic} distributed algorithm $A'$ that guarantees that for at least half of pairs $P=(u,v)\in \mathcal P $ we will be able to find a feasible $(u,v)$-path using $A'$.
\end{theorem}
Both this and the previous theorems give certain deterministic algorithms, although the last theorem only guarantees its existence. Constructing such routing tables could be made algorithmically efficient using standard techniques of derandomization that we avoid here.
\begin{proof} We first define a randomized algorithm $A'$ as follows. We first form a composite metric $w = w_1+w_2$ on each edge. Assume that a packet is at vertex $x$ and is being routed to a vertex $v$. For $x$, we have a list $\mathcal S_v:=\{S_1,\ldots, S_k\}$ of different paths to $v$ such that $\mathcal S$ consists of all paths that have minimum weight in $w$-metric. Then we check the weights $(w_1,w_2)$ of these paths. If one of these paths has positive weights in both coordinates then we route the packet to the next hop on this path. If all paths have $0$ weight in the first coordinate then we route along any of them. We do the same if all paths have $0$ weight in the second coordinate. If there are both paths with $0$ in the first coordinate and $0$ in the second coordinate then we choose an element $i\in \{1,2\}$ uniformly at random and route the packet along (one of) the path that has $0$ weight in the $i$-th coordinate.

Let us analyze the situations when such a randomized algorithm fails to find a feasible path. Let $R$ be a found $(u,v)$-path, where $P =(u,v)\in \mathcal P$, and assume that $R$ is not feasible. Clearly, the composite weight of $R$ is at most $2a_P$, since we assumed that there is a path of weight at most $(a_P,a_P)$, which has composite weight $2a_P$. Thus, $R$ must have weight of one of the forms: $(0,2a_P)$ or $(2a_P,0)$. Without loss of generality, assume that it has the form $(0,2a_P)$. Assume that $x$ is the first vertex on the path such that  all paths in $\mathcal C_x$ have weights with $0$ in one coordinate. First, we note that $x\ne u$ since there is a feasible path from $u$ to $v$.
Thus, the part of $R$ to $x$ must have weight $(0, \alpha)$ for $\alpha>0$. Second, since there is a feasible path from $u$ to $v$, the set $\mathcal C_x$ must contain a path with weights of the form $(2a_P-\alpha, 0)$. Indeed, this is by the definition of $x$: the vertex $x'$ just before $x$ on $R$ had a path in $\mathcal C_{x'}$ with both positive weights. But if all paths in $\mathcal C_x$ have weight $(0,2a_P-\alpha)$ then the weights of all paths in $\mathcal C_{x'}$ through $x$ have $0$ weight in the first coordinate, a contradiction. Given that $\mathcal C_x$ has a path with weight $(2a_P-\alpha,0)$, we had a $1/2$ chance of choosing it using our randomized next hop algorithm, which means that at this step we had a probability at least $1/2$ of choosing not $R$ but a feasible path. Thus, the probability of choosing an infeasible path $R$ is at most $1/2$.

Summarizing, for any given $P =(u,v)\in \mathcal P$, the probability of finding a feasible path for $P$ using $A'$ is at least $1/2$. Using linearity of expectation, the expected number of feasible paths that $A'$ generates, is at least $\frac 12 |\mathcal P|$. Therefore, there exists a particular deterministic choice of next hop routing rules that results in at least that many feasible paths.

\end{proof}

We can consider the following modification of the construction in the Figure~\ref{not100-2}, showing that not only 50\% discovery rate,
but even 1\% is impossible in general case, although the network becomes less natural.

Consider the following network $G_k = (V,E)$ with two costs on each edge: $V = \{c,v_1,\ldots, v_{kn}\}$
and $E = \cup_{i=1}^{kn} E_i$, where $E_i$ consists of $k+1$ edges connecting $c$ and $v_i$:
$k$ edges from $c$ to $v_i$ with costs $(j,k-1-j)$ for $j=0,\ldots, k-1$, respectively,
and one edge form $v_i$ to $c$ with cost $(\ell,k-1-\ell)$ if $\ell n < i \le (\ell+1)n$.

Similarly to Theorems \ref{th2} and \ref{th3}, we can prove the following result.
\begin{theorem}\label{th4}
 Assume that $A$ is a distributed routing algorithm that aims to find paths between any two pairs of vertices in $G_k$ that satisfy requirements $(w_1,w_2)< (k,k)$. Then

 \begin{itemize}
   \item[(i)] There is a satisfying path between any two vertices
   \item[(ii)] the algorithm $A$ cannot find satisfying paths for more than $(kn+1)^2-kn((k-1)n-1)$ ordered pairs of vertices out of $(kn+1)^2$. That is, asymptotically we can find satisfying paths for no more than $1/k+o(1)$ fraction of pairs.
 \end{itemize}
\end{theorem}

By making $k$ large we can make the discovery rate arbitrarily small.

We also note that we can avoid using multiple edges in $G_k$ by `duplicating' each vertex $v_i$, one copy per edge. The orientation seems to be more important, although similar, but weaker, examples, are possible even without orientation.

It is also interesting to figure out, what is the dependence between the discovery rate and the `slack' that the restrictions leave. Are the examples above are optimal in that sense?

\section{Linear combination of weights}
We have shown that  in general case it is impossible to guarantee  distributed routing with $100\%$ discovery rate under several constraints
even if a feasible path is present for every source and destination pair.

In this section, we are going to analyze the {\it linear} class of algorithms that, for, say, two weights $W_1(e)$ and $W_2(e)$, assign a weight $pW_1(e)+(1-p)W_2(e)$ to the edge $e$, where $p\in [0,1]$. The case of more weights is analogous.

Assume that the costs, or weights, $W_i(e), 1\leq i\leq k$  are assigned to the graph edges $e$.
Finding multi-constrained paths for constraints $C_i$ on the $i$-th metric amounts to finding paths that satisfy
$$
\sum_{e\in path} W_i(e) < C_i, \quad \forall i,
$$
for each path.

A pair of nodes is called {\it satisfying} if such a path is found by the algorithm,
{\it non-satisfying} if we prove that such a path does not exist, and {\it uncertain} if the algorithm can not decide whether such a path exists or not.

 Varying parameter $p$ within $[0,1]$ and considering corresponding shortest path in the composite metric does not always allow to find a path that satisfies constraints even if it exists. The following example is similar to the one that has already been mentioned in \cite{Puri}.
For the weight$\, =p W_1+(1-p) W_2$, see Figure \ref{not100case}:
\begin{itemize}
\item $p>0.5$, top path is the shortest,
\item $p=0.5$, both top and bottom paths give the same value and are the shortest,
\item $p<0.5$, bottom path is the shortest.
\end{itemize}

Middle path, which is the only one that satisfies both constraints $\leq 1$, would never be found.

\begin{figure}[!h]
\centering
\begin{tikzpicture}
  \coordinate (O) at (0,0) node [left] {s} ;
  \coordinate (A) at (3,0) ;

  \fill (O) circle (2pt);
  \fill (A) circle (2pt);

  \draw[] (O)-- node [above] {$0.9,0.9$} (A);

  \draw[] (O) to [bend left=70] node [above] {$0.1,1.1$}  (A);
  \draw[] (O) to [bend right=70]node [below] {$1.1,0.1$} (A);

\end{tikzpicture}
\caption{Middle path satisfies constraint would never be found by SPF with a linear composite metric.}

\label{not100case}
\end{figure}
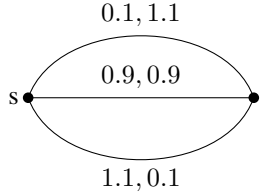

\subsection{Optimal parameter for a composite metric}

A well-known approach is to consider a linear combination of costs as a new edge weight and find the corresponding shortest paths \cite{Jaffe,Puri}.
Deploying the shortest path algorithms for a composite metric $p_1 W_1+p_2 W_2+\ldots + p_k W_k$, one would like to expect that as many as possible found shortest paths would be feasible for all constraints.

First, we establish some general theoretical results to illustrate the idea and develop useful tools. We show that some proportions are better than  others in that particular sense. We propose to maximize the probability that a path satisfies the constraints while the path is the shortest for metric $p_1 W_1+p_2 W_2+\ldots + p_k W_k$ as a function of parameter ${\bf p}=(p_1,\ldots,p_k), p_i\geq 0, \sum_i p_i=1$.

 We provide theoretical results for the following ``general transmitting scheme'' graph, Figure \ref{general}, where each edge between source and destination is equipped with some random costs: $\mbox{weight}_1=X_i$, $\mbox{weight}_2=Y_i,\ldots, \mbox{weight}_k=Z_i$ where $1\leq i \leq n$. Thus, one can choose optimal ${\bf p}$ in the sense of maximising $P(A)$ where $A$ is ``shortest path satisfies the constraints''. Obviously, by the law of total probability
 \begin{equation}
 P(A)= \sum_{i=1}^{n}P(AB_i)
 \label{totalprobab}
 \end{equation}
 where $B_i$ is ``$i$-th path is the shortest for metric $p_1 W_1+p_2 W_2+\ldots + p_k W_k$''. All the summands in (\ref{totalprobab}) are equal due to symmetry and maximizing one summand means maximizing $P(A)$ .

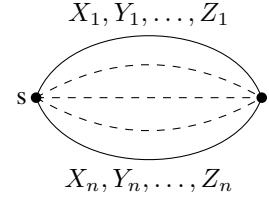
\begin{figure}[!h]
\centering
\begin{tikzpicture}
  \coordinate (O) at (0,0) node [left] {s} ;
  \coordinate (A) at (3,0) ;

  \fill (O) circle (2pt);
  \fill (A) circle (2pt);

  \draw[dashed] (O)--(A);

  \draw[dashed] (O) to [bend left=30] (A);
  \draw[dashed] (O) to [bend right=30] (A);
  \draw[] (O) to [bend left=70] node [above] {$X_1,Y_1, \ldots, Z_1$}  (A);
  \draw[] (O) to [bend right=70]node [below] {$X_n,Y_n, \ldots, Z_n$} (A);

\end{tikzpicture}
\caption{``General transmitting scheme'' graph}
\label{general}
\end{figure}

\begin{theorem}\label{thp}
Let $\{X_i\}, \ldots, \{Z_i\}, 1\leq i \leq n,$ be independent random variables
from $k$ absolutely continuous distributions on the segments with bounded densities.
Then there exists ${\bf p}=(p_1,\ldots,p_k)$ that maximize the probability
$P(A)$.
\end{theorem}

\begin{proof}
For the sake of simplicity we consider $k=2$ that corresponds to $2$-dimensional case, see Figure \ref{slope} representing the joint
probability space. Paths $(X_i,Y_i)$ are represented by black dots with the corresponding coordinates.
The oblique line $p_1 x+p_2 y= const$ which corresponds to the cost of the shortest path,
and for sure $(x_{sh},y_{sh})$, the shortest path itself, belongs to that line. Its slope depends on $p_1, p_2$.
We should note that which path becomes the shortest one depends on $p_1,p_2$ of course.
If some path is not the shortest one, then $p_1 x+p_2 y > p_1 x_{sh}+p_2 y_{sh}$, and the corresponding black dot lies above
the oblique line.
For a general $k$-dimensional case we should modify the Figure \ref{slope} and use a hyperplane instead of the oblique line.

\begin{figure}[h]
\centering
\includegraphics[width=7cm]{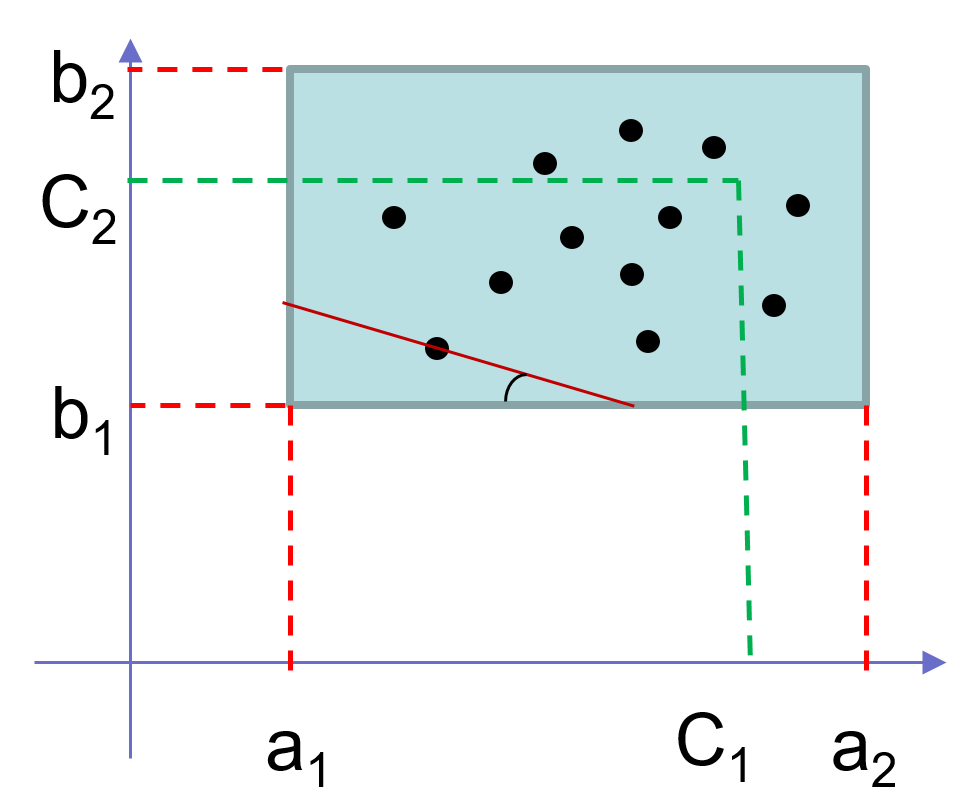}
\caption{ Example for $k=2$. Costs $W_1$ and $W_2$ are randomly taken from the intervals $[a_1, a_2]$ and $[b_1, b_2]$.
	The constraints $C_1$ and $C_2$ bound the area the feasible paths should belong.} \label{slope}
\end{figure}

In order to maximize the number of feasible paths the algorithm can discover for a given ${\bf p}=(p_1,p_2)$,
we should maximize the probability $P(A B_1)$
that (1) path satisfies the constraints and (2) the path is the shortest one for the metric $p_1 W_1+p_2 W_2$,
i.e.
$$
P(A B_l)= \int_{a_1}^{C_1} \int_{b_1}^{C_2} (1-pr(x,y;p))^{n-1} dG_1(x) dG_2(y),
$$
where
$pr(x,y;{\bf p})$ is the probability for the path to find itself in the
left bottom triangle bounded by the oblique line
$$
p_1 x+p_2 y= const(shortest\ path).
$$

Clearly, due to the bounded probability densities it is a continuous function of parameter ${\bf p}\in (k-1)$-dimensional simplex.
Thus, it attains maximum at some point of a compact set.
\end{proof}

If seems to be obvious that if some path is the shortest one for a metric $p_1 W_1+p_2 W_2$
and we normalize $p_1$ and $p_2$ as $p$ and $p-1$, where $p=\frac{p_1}{p_1+p_2}$,
the path still be the shortest one for the metric $p W_1 + (1-p) W_2$.
In addition, for a $2$-dimensional case with uniform distribution of costs
$G_1[a_1,a_2]$ and $G_2[b_1,b_2]$,
if the constraints $C_1, C_2$ are sufficiently small in comparison to $a_2$ and $b_2$,
then $p$ can be easily calculated. Optimal $p$ corresponds to such an oblique line
in the Figure~\ref{slope} that maximizes the probability of the pentagon, i.e. minimizes
the area of the triangle, i.e.
$$
\frac{p}{1-p} = \frac{C_2-b_1}{C_1-a_1}\ ,
$$
which can be rewritten as
\begin{equation}
p=\frac{C_2-b_1}{C_2-b_1+ C_1-a_1} \label{bestp}.
\end{equation}
The parameter $p$ gives maximum discovery rate of feasible paths by a shortest path
search algorithm with a linear weighted combination of costs.
It should be noted that no guarantee for the distributed routing discovery rate is possible
(Theorems \ref{th1}, \ref{th2}).
Multi-constrained distributed routing algorithms have to use some additional information
(e.g. to find out where the packets came from) in order
to make its routing decision.

It should be noted that the existence of the best parameter for a linear combination of weights
has been derived under several assumptions of Theorem~\ref{thp}. The assumption of identical distribution of path costs $X_1$, $X_2$, etc. is not really important.
Although it allow us to make calculations for the Theorem~\ref{thp} and (\ref{bestp}) in a more transparent way.
The least natural one is the assumption of independence of path costs $X_1$, $X_2$, etc.
Even if costs of graph edges are random and independent for some scenario,
hardly can we expect that all the shortest paths are disjoint.
Shortest paths from some source node to different destinations most likely to have
some edges in common and therefore we should not assume the path costs $X_i$ completely independent.
Nevertheless, Theorem~\ref{thp} gives us a motivation to expect "good behavior"
of the probability $P(A)$ as a function of parameter ${\bf p}$, which could at least simplify our search
for an optimal value of ${\bf p}$ via an iterative algorithm.

\subsection{Concavity}



For some source and destination vertices, let $x_{sh}({\bf p})$, $y_{sh}({\bf p})$, $\ldots$, $z_{sh}({\bf p})$ be the weights of the shortest path that minimize
\begin{equation}
p_1 W_1+\ldots + p_k W_k.
\label{metric}
\end{equation}
Let us analyze, what kind of function this minimum is. For each fixed path that has costs $x,y,\ldots,z$, the function $p_1x+p_2y+\ldots+p_kz$ is a linear function of $\mathbf{p}$, whose graph is a hyperplane. The minimum \eqref{metric} is a pointwise minimum of such functions over all possible paths between two vertices. In other words, the set under the graph of the minimum \eqref{metric} is an intersection of several halfspaces. Each facet (or segment in 2-dimensional case) corresponds to a particular  path, while intersections of facets (vertices in 2-dimensionnal case) correspond to two different paths with different values of metrics, whose values of the composite metric coincide.



\begin{figure}[tp]
\centering
\includegraphics[scale=0.7]{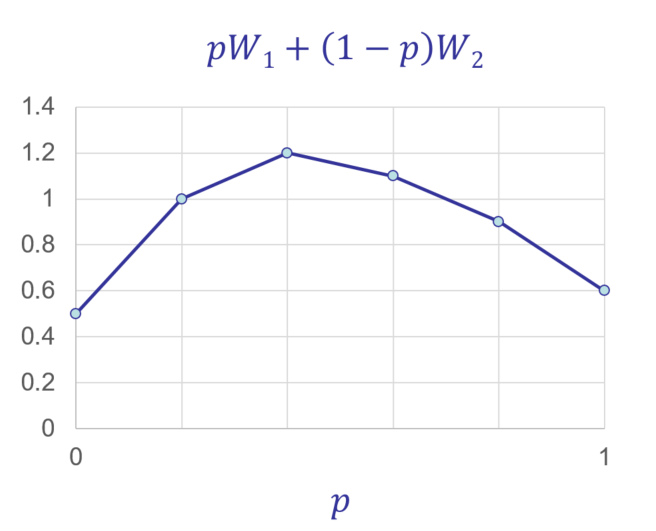}
\caption{Example of a function $p x_{sh}+(1-p) y_{sh}$. Positive slope means that $x_{sh}>y_{sh}$, negative that $x_{sh}<y_{sh}$, zero slope that $x_{sh}=y_{sh}$. Theorem \ref{thmconcave} states it must be concave.}\label{concave}
\end{figure}

\begin{theorem}
Let ${\bf p}\in unit\ (k-1)-simplex$.
The following function of ${\bf p}$:
$$p_1 x_{sh}({\bf p})+\ldots +p_k z_{sh}({\bf p})$$ is concave.
\end{theorem}\label{thmconcave}
\begin{proof}
First we consider the case of $n=2$.
Let the path $(w_1,w_2)$ be the shortest one for $p\in segment$ and the path $(v_1,v_2)$ be the shortest one
for $p\in next\ segment$ (Figure~\ref{notconcave}). If graph is not concave down then the path
$(v_1,v_2)$ can not be the shortest one for $p\in next\ segment$ due to inequality
$p v_1+(1-p) v_2>p w_1+(1-p) w_2$ (red dotted line in the Figure~\ref{notconcave}).

Thus, an increase in the value of $p$ leads to a decreasing line slope which corresponds to the ``next segment'' shortest path.

The multidimensional case can be reduced to $2$-dimensional. Namely, consider any two points from the graph of the function $p_1 x_{sh}({\bf p})+\ldots +p_k z_{sh}({\bf p})$. Consider the restriction of this function to the segment that connects our two points. Then for this restriction we get the same picture as shown in the Figure~\ref{concave}, and the restriction is concave. This implies that the function itself is concave.
\end{proof}
\begin{figure}[tp]
\centering
\includegraphics[scale=0.4]{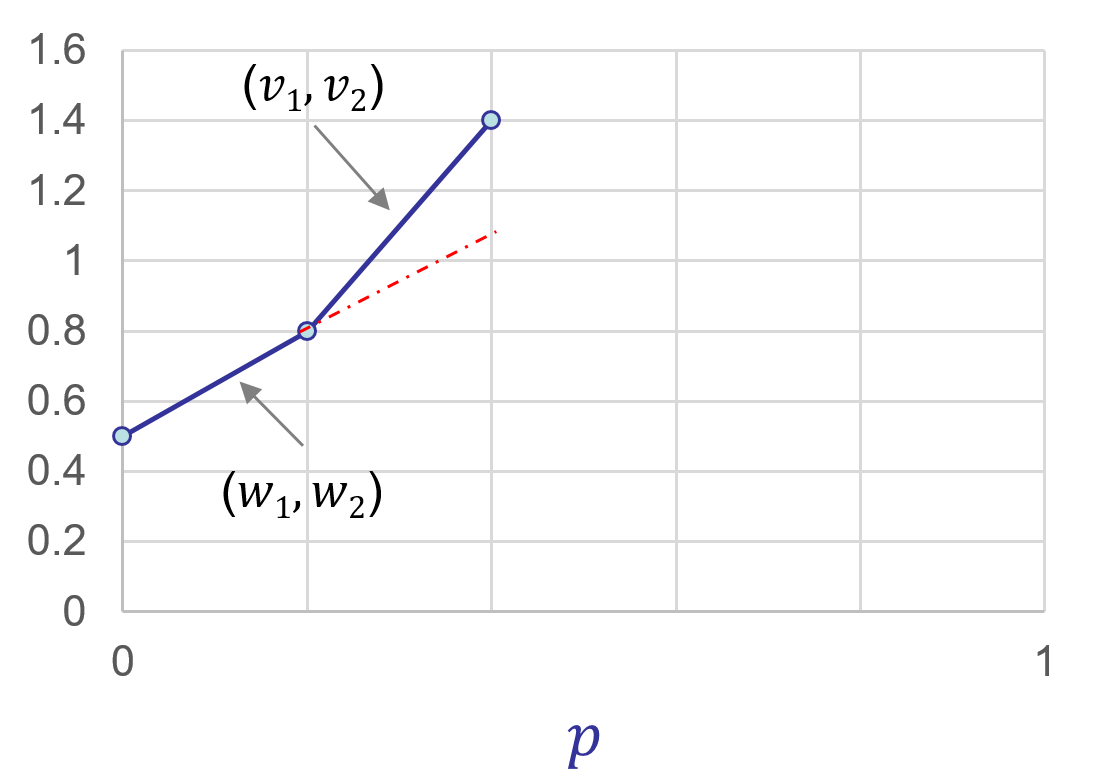}
\caption{Non-concavity leads to a contradiction.}\label{notconcave}
\end{figure}

In the next section we show that if all constraints are of the form $\leq 1$ and if $$p_1 x_{sh}+\ldots +p_k z_{sh}>1$$ for some ${\bf p}$ (as on the Figure (\ref{concave}), then there is no feasible path between the source and destination.

Consider shortest paths $(x_{sh},y_{sh})$ and $(x'_{sh},y'_{sh})$ for different values of $p$ for a $2$-dimensional case.
Note, that there are only two possibilities: $x_{sh}\leq x'_{sh},\ y_{sh}\geq y'_{sh}$ or $x_{sh}\geq x'_{sh},\ y_{sh}\leq y'_{sh}$,
otherwise one path is always shorter than the other one in the sense of the metric (\ref{metric}).

Since the slope of the line $p x_{sh}+(1-p) y_{sh}$, as a function of $p$, is equal to $x_{sh}-y_{sh}$,
Theorem~\ref{thmconcave} implies the following result which has been stated in \cite{Puri}:
\begin{corollary}
Let paths $(w_1,w_2)$ and $(v_1,v_2)$ be the shortest paths between some nodes $A,B$ in the sense of $p W_1+(1-p) W_2$
for  $p=p'$ and $p=p'',$ respectively. If $p'<p''$, then $w_1\geq v_1$ and $w_2\leq v_2$.
In other words, increasing the value of $p$ should decrease $W_1$ and increase $W_2$ for the new shortest path.
\label{monoton}
\end{corollary}

This could be very useful when searching for an optimal $p$, although it works only in the $2$-dimensional case.


\section{Iterative algorithm}

Though no algorithm can guarantee 100\% discovery rate of feasible paths,
an optimal linear combination
of weights could be found to maximize the feasible paths discovery rate for distributed routing.
Similar iterative techniques were previously used to find such an optimal combination, see e.g. \cite{Puri}.

We show that though finding the maximum number of feasible paths is equivalent to a global maximum search
for a non-convex function it is still possible to find it as if it were convex.

Let $N_y$ be the number of {\it satisfied} pairs, i.e. a feasible path can be found by our algorithm for every such pair of source-destination nodes.
Let $N_n$ be the number of {\it non-satisfied} pairs, i.e. we can prove that no feasible paths exist for such pairs of nodes.
The rest of the pairs are called {\it uncertain} since the algorithm can not decide whether a feasible path exists or not for such pairs.
If $N_{tot}$ is the total number of pairs,
the number of uncertain pairs is $N_u = N_{tot} - N_y - N_n$.
Our goal is to maximize the number of satisfied pairs.
Minimizing the number of uncertain pairs could be another possible option.

First of all, we normalize all the costs by the corresponding constraint values.
The condition for the total costs along feasible paths becomes $ \sum w_i < 1 \ ,$
where $w_i = W_i / C_i$.
It seems to be obvious that the shortest path tree remains the same after such a normalization.

Then we normalize the parameter ${\bf p}$ to ensure $\sum p_k = 1$. It is clear that if some path is the shortest one
in the sense of the metric $p_1 w_1 + p_2 w_2$, for example, the same path remains to be the shortest one
for the metric $p w_1 + (1-p) w_2$.

We are searching for the global maximum of $N_y$, which is a function of the parameter $p$.
If every router in the network searches the maximum deterministically, e.g. by
a predefined grid search, some gradient method or some iterative method with a number of iterations fixed in advance,
 etc., every router finds the same value of $p$
without any information exchange between the routers.
For the same value of $p$, the same scalar metric for every router makes a distributed routing possible.
The absence of loops is inherited from single metric shortest path algorithms.

A special note should be made regarding non-satisfied pairs, $N_u$.
If we compute the shortest path tree for some source node $s$, and for some given value of the parameter, $p'$,
we have two following options (Figure~\ref{satnonsat}).
If for some pair of nodes $\sum w(p') < 1$, where $w = p' w_1 + (1-p') w_2$,
we need to check if the correspondent path is feasible,
i.e. if both $\sum w_1 < 1$ and $\sum w_2 < 1$.
However, if for some pair of nodes $\sum w(p') > 1$ then no feasible path could exist between these nodes,
even for a different value of the parameter $p$.
We can prove it as follows.
Let $\sum w = \sum \left(p' w_1 + (1-p') w_2\right)> 1$ along the shortest path between some nodes.
Let's assume that a feasible path do exists between these nodes.
Along that path $\sum  \left(p' w_1 + (1-p') w_2\right) = p' \sum  w_1 + (1-p') \sum w_2$.
Since $\sum w_1 < 1$ and $\sum w_2 <1$ along that path, we can conclude that
$\sum w < p' + (1-p') = 1$. Since we came to the contradiction, we proved that
if for some value of $p'$
\begin{equation}
 \sum \left(p' w_1 + (1-p') w_2\right)> 1 \Longrightarrow \text{ non-satisfied} \ ,
\label{N_n}
\end{equation}
i.e. no corresponding feasible path could exist.
It should be noted that in the case of $\sum w(p) < 1$ we can not make such a definite conclusion
since for a different value of $p$ that inequality might not hold.

\begin{figure}[tp]
\centering
\includegraphics[scale=0.25]{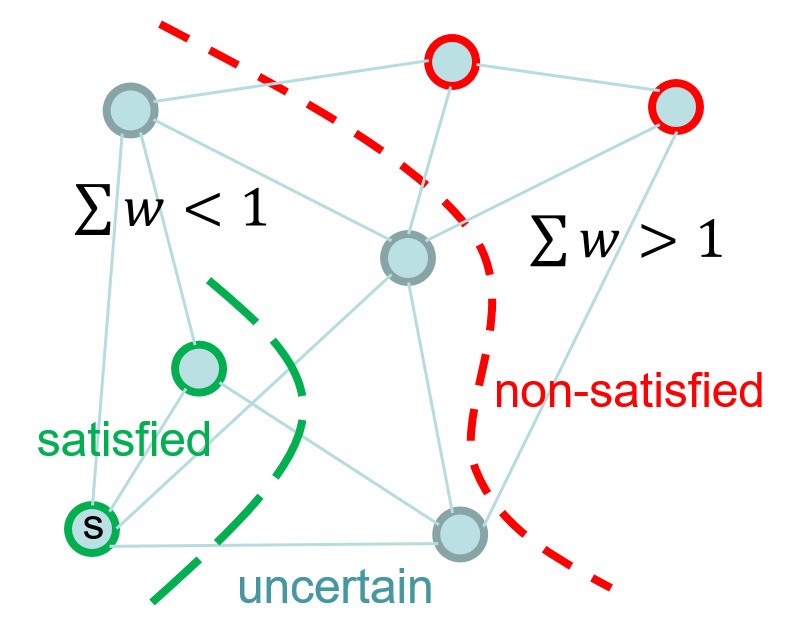}
\caption{There are two options for every path from one-to-all shortest path tree
with metric $w = p w_1 + (1-p) w_2$. If $\sum w > 1$ for some pair of nodes then no feasible path could exist
between these nodes. If $\sum w < 1$ we still need to check if the correspondent path is feasible,
i.e. if both $\sum w_1 < 1$ and $\sum w_2 < 1$.}
\label{satnonsat}
\end{figure}

\begin{figure}[tp]
\centering
\includegraphics[scale=0.37]{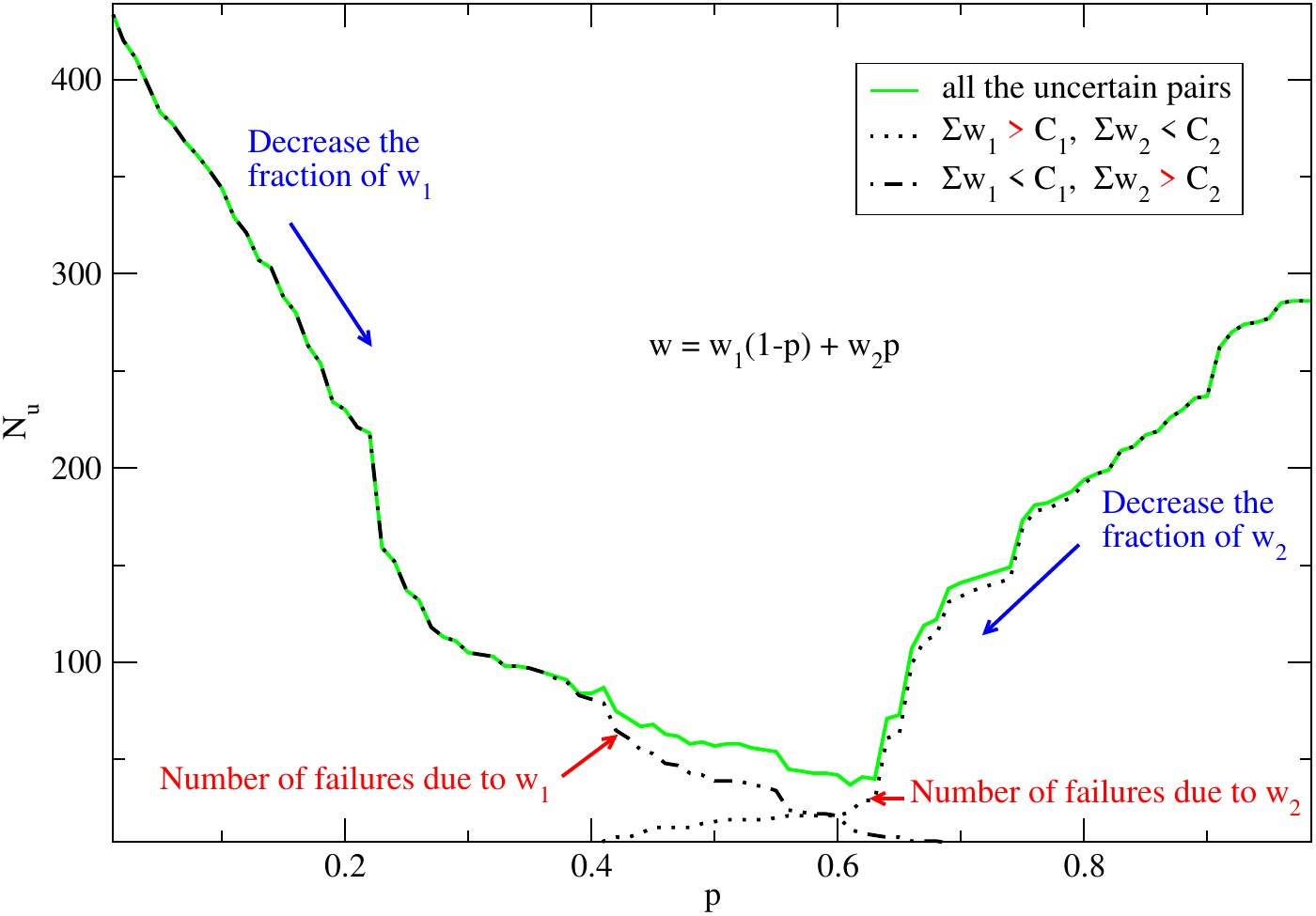}
\caption{Typical number of uncertain pairs (green solid line) along iterative search using $w=w_1(1-p)+w_2 p$.
Number of failure functions (black dotted lines) are monotone due to the Corollary \ref{monoton}.}
\label{typical_p}
\end{figure}

Figure~\ref{typical_p} presents a typical behavior of the number of uncertain pairs $N_u$ as a function of $p$.
In agreement with Theorem~\ref{thp}, $N_u(p)$ shows relatively "good behavior", though it is non-monotone.
While searching for the minimum of $N_u(p)$ and computing shortest path tree for some fixed value of $p$,
we could check the constraint satisfaction and "prune" the unfeasible branches.
 Corollary \ref{monoton} suggests the following procedure. We count the "pruned" pairs with pruning caused by the violation of
the first and the second constraint separately. We change $p$ in order to
equalize these numbers.
For example, if most of the paths are pruned due to the second constraint,
i.e. $$\sum w_1 < 1, \  \sum w_2 >1, \ \sum((1-p)w_1 + p w_2) < 1$$ we should decrease $p$.
In the opposite case, if $\sum w_1 > 1, \sum w_2 <1$, we decrease $(1-p)$ by increasing the value of $p$.
This property allows us using global minimum search by any algorithm designed for monotone functions
rather than using more complex algorithms for a global minimum search.

Our iterative algorithm overview:

\noindent 1. Fix $p\in [0,1]$ and set $w_p=pw_1+(1-p)w_2$\\
2. Run SPF with $w_p$ as a weight function\\
3. Repeat to find the optimal $p$ using monotonicity.

It should be noted that the number of satisfied pairs, $N_y(p)$, could be found relatively quickly (Figure~\ref{conv}).
On the other hand, the number of non-satisfied pairs, $N_n(p)$,
is cumulative and therefore is monotonically increasing after each iteration
in accordance with the condition~\ref{N_n}.
If we are mostly interested in satisfied pairs, a faster, e.g. Newton's, algorithm could be used.
If our goal is to minimize the number of uncertain pairs then we should rather prefer a slower convergence
and “more uniform” sampling for $p \in [0,1]$ like in a golden section search (Figure~\ref{conv}).

\begin{figure}[tp]
\centering
\includegraphics[scale=0.37]{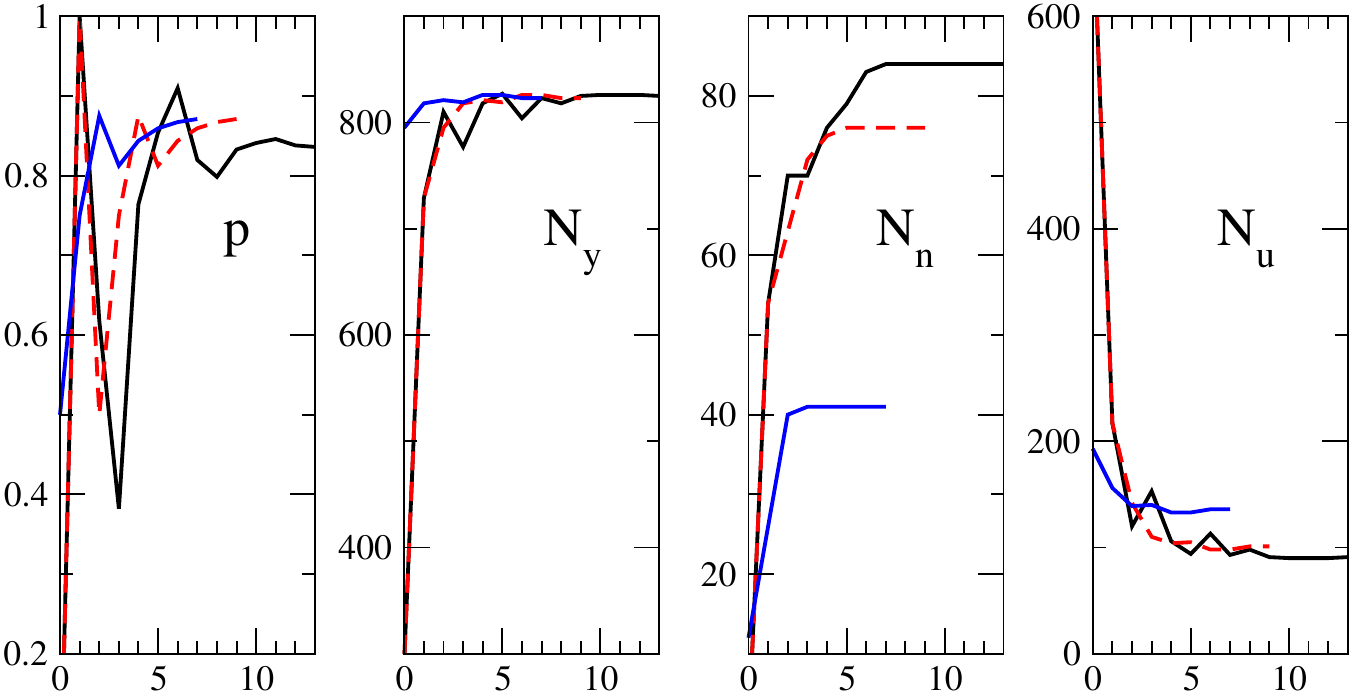}
\caption{Typical behavior of $p$, number of satisfied ($N_y$), non-satisfied ($N_n$), and uncertain ($N_u$) pairs
versus the number of iterations for different algorithms: golden section search (black solid), and dichotomy with (red dashed)
and without (blue dotted) checking the constraints at $p=0$ and $p=1$.
At every iteration, shortest path tree is computed for a fixed value of $p$.
}
\label{conv}
\end{figure}

We suggest using an iterative dichotomy algorithm (red line in the Figure~\ref{conv}). It converges to the optimal $N_y$
as well as accumulates $N_u$ relatively quickly due to the additional checks of the condition (\ref{N_n}) at $p'=0$ and $p'=1$.
It should be noted that it is not necessary to store the lists of feasible or unfeasible pairs.
To find the optimal parameter $p$ we need to count the number of such pairs only.

\begin{algorithm}[hpt]
\caption{}
\begin{algorithmic}[1]
\State Normalize costs $W_i$.
\State List \{ for all (A,B) $\rightarrow$ uncertain\}  //  {\it list of all possible pairs, all are uncertain by default}
\While {in time or discovery rate < sufficient level}
\For{p in $\{0;1;\frac{1}{2};\frac{1}{4} \mbox{or} \frac{3}{4} ; \ldots \mbox{etc.} \}$}
  \For{all A in Nodes}
    \For{all B $\neq$ A in Nodes}

  \State  Find shortest paths \{(A, to all)\ |\ p\} // using edge weight $:=p w_1+(1-p) w_2$

  \If{(A,B)$|_p$ satisfies constraints}
  \State update List\{(A,B) $\rightarrow$ there exists path that satisfies\}
  \EndIf

    \If{(A,B)$|_p$: $p w_1+(1-p) w_2>1$}
    \State update List\{(A,B) $\rightarrow$ there exists no path that satisfies\} //  {\it if it exists, then $p w_1+(1-p) w_2\leq 1$ for such path}
    \EndIf

    \EndFor
  \EndFor
\EndFor
\EndWhile

\end{algorithmic}
\end{algorithm}

Several options are possible. In a distributed mode we could choose one optimal value of $p$, same for every node.
This can be done if all the nodes follow the same iterative procedure and the accuracy, or the number of iterations,
is predefined in advance. Since some values of $p$ could be more "favorable" for some nodes than for the others,
we should possibly compute all-to-all paths in order to select the optimal $p$ for all the nodes.
On the other hand, in a centralized mode we could possibly broadcast the value of $p$ precomputed at a powerful
central node.

It could also be possible to use several values of $p$ simultaneously, e.g. for different types of traffic.
Though there is a probability that the shortest paths for different $p$ be partially disjoint
the sizes of the routing tables have to be multiplied respectively.
On the other hand, the cumulative number of feasible paths could increase due to the usage of multiple values of $p$.

It should be noted that though we compute an optimal $p$ that maximizes the total number of all-to-all satisfied pairs
it is also possible to select, maybe randomly though deterministically, some 'special' node.
We can count one-to-all satisfying paths from that node and choose the value for the optimal $p$ as one that
maximizes the number of these paths. If edge costs are distributed more or less randomly in the network
the obtained 'one-to-all optimal $p$' will be very close to 'all-to-all optimal $p$.
The choice of the 'special' node should be deterministic, so every node obtains the same one and therefore
every node computes the same optimal 'one-to-all $p$' value.

\section{Results for specific topologies}

We studied our algorithm performance for different graph topologies:
grid 15x15 and 45x45 for two cases when all-to-all or border-to-border pairs were considered;
dual-home topology (1000 nodes); mouth-like topology (1000 nodes).
Core of dual-home and mouth-like networks consisted of 10 ordered pairs of nodes $a$ and $b$.
Nodes of each pair had an edge between them, all $a$-nodes formed a clique as well as
all $b$-nodes were fully connected. The rest of the nodes were equally distributed among
the pairs. Each node was connected to both $a$-node and $b$-node of core pair for the
dual-home network. As for the mouth-like network, the rest of the nodes were equally distributed
in pairs among the pairs of the core. Each such pair was connected to both $a$-node and $b$-node of core pair.

Costs were randomly assigned: $W_1$ - from a normal distribution with 7.5 mean and 1.25 variance and
$W_2$ - from a discrete uniform distribution {0.01, 0.02, 0.03, 0.04, 0.05}. Cost $W_1$ represented latency,
cost $W_2$ represented an additive measure of packet loss rate.
For a three dimensional case we considered an additive measure for jitter $W_3$
from a positive normal distribution with 2.0 mean and 2.0 variance.

To explore the influence of the constraints we set the constraint $C_1$
to the maximum possible delay for the longest shortest path for the metric $w = W_1$ in our network.
Similarly, $C_2$ and $C_3$ were set to maximum possible packet loss rate, and jitter, on the longest shortest path in our network
for the metric $w = W_2$, and $w = W_3$, respectively.
In non-periodic grid topology, for example, the longest shortest path would normally be the
one connecting the opposite corners of the grid.
In addition, we multiply these maximum possible constraints by a coefficient $\alpha \in [0.5,1]$.
The larger the coefficient $\alpha$, the larger the number of node pairs for which a feasible path could possibly exist.
The case of $\alpha \ge 1$ corresponds to an unconstrained routing.
We count the number of satisfied, non-satisfied, and uncertain all-to-all or border-to-border pairs
as a function of $\alpha$.
In addition, we compute a discovery rate which was defined as
$$
\mbox{R} = \frac{N_y}{N_y+N_u}.
$$
Averaged over 10 random realizations results for the topologies mentioned are shown in the Table~\ref{table1}
and in the Figures~\ref{r1}, \ref{r2}, and \ref{r3}.

\begin{table}
\centering
\begin{tabular}{c|c|c|c}
a2a & $N_{tot}$ & $N_u^{max}$, \% & $R^{min}$, \% \\
\hline
$DH $	& $10^6$ 		& 5.2--7.5 [10.3--14.0]& 86--86 [75--72]\\
$ML $	& $10^6$ 		& 4.4--7.1 [7.7--11.3]& 94--92 [88--85]\\
$G15$	& 50625 		& 1.2--2.5 & 99--96 \\
$G45$	& 4100625 		& 0.7--2.3 [1.4--4.4]& 99--96 [97--92]\\
\hline
b2b&		&  		  & \\
\hline
$G15$	& 3136 			& 1.9--4.0 & 97--94 \\
$G45$	& 30976 		& 1.4--4.5 [2.1--7.4]& 98--94 [97--90]\\
\hline
\end{tabular}
\caption{Total number of pairs $N_{tot}$, maximum number of uncertain pairs $N_u^{max}$, minimum discovery rate $R^{min}$,
for all-to-all (a2a) and border-to-border (b2b) topologies:
$DH$ (Dual-home), $ML$ (Mouth-like), $G15$ (Grid 15x15), $G45$ (Grid 45x45).
Two values are given for $N_u^{max}$ and $R^{min}$: the left one - for multiple values of $p$, the right one -
if a single value of $p$ is used. Two weights scenario. The results for three weights are shown in square brackets.
}
\label{table1}
\end{table}

\begin{figure}[tp]
\centering
Two weights:
\includegraphics[scale=0.36]{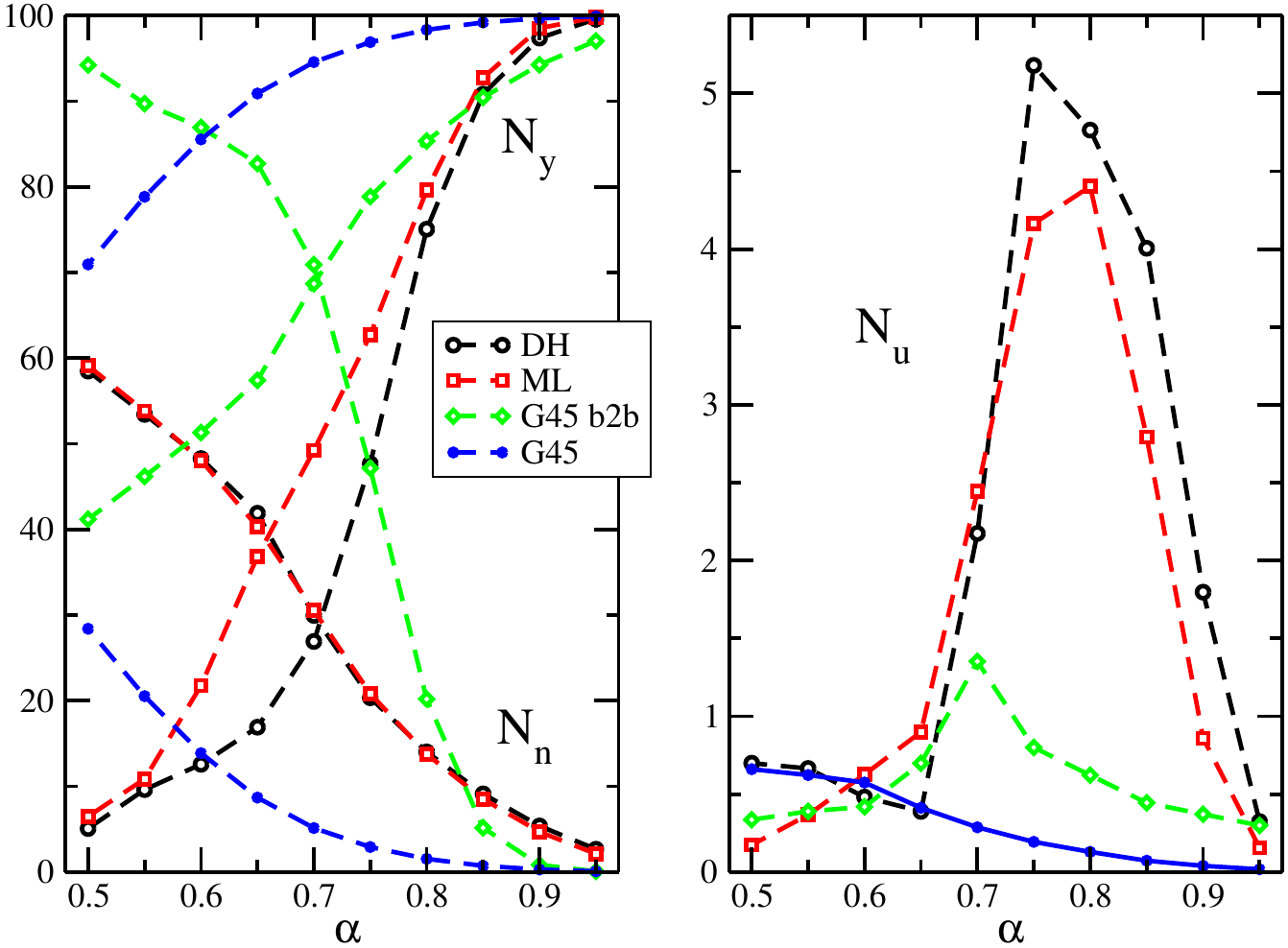}\\
Three weights:
\includegraphics[scale=0.36]{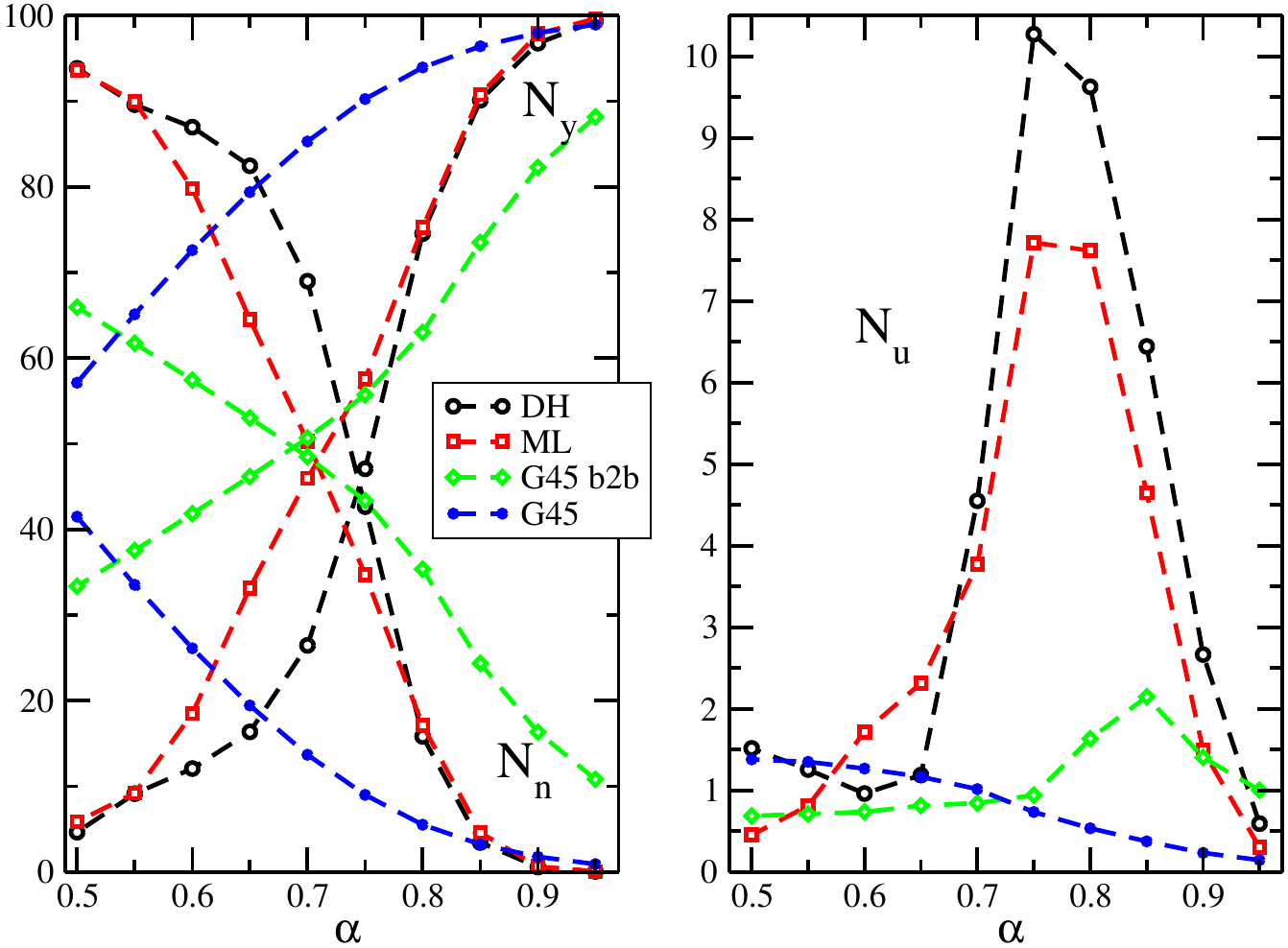}
\caption{"Multiple values of $p$ mode" routing. Satisfied, $N_y$,  non-satisfied, $N_n$, and uncertain, $N_u$,
pairs in percents as functions of the coefficient $\alpha$ for different topologies: dual-home, mouth-like, and grid.}
\label{r1}
\end{figure}

\begin{figure}[tp]
\centering
Two weights:
\includegraphics[scale=0.36]{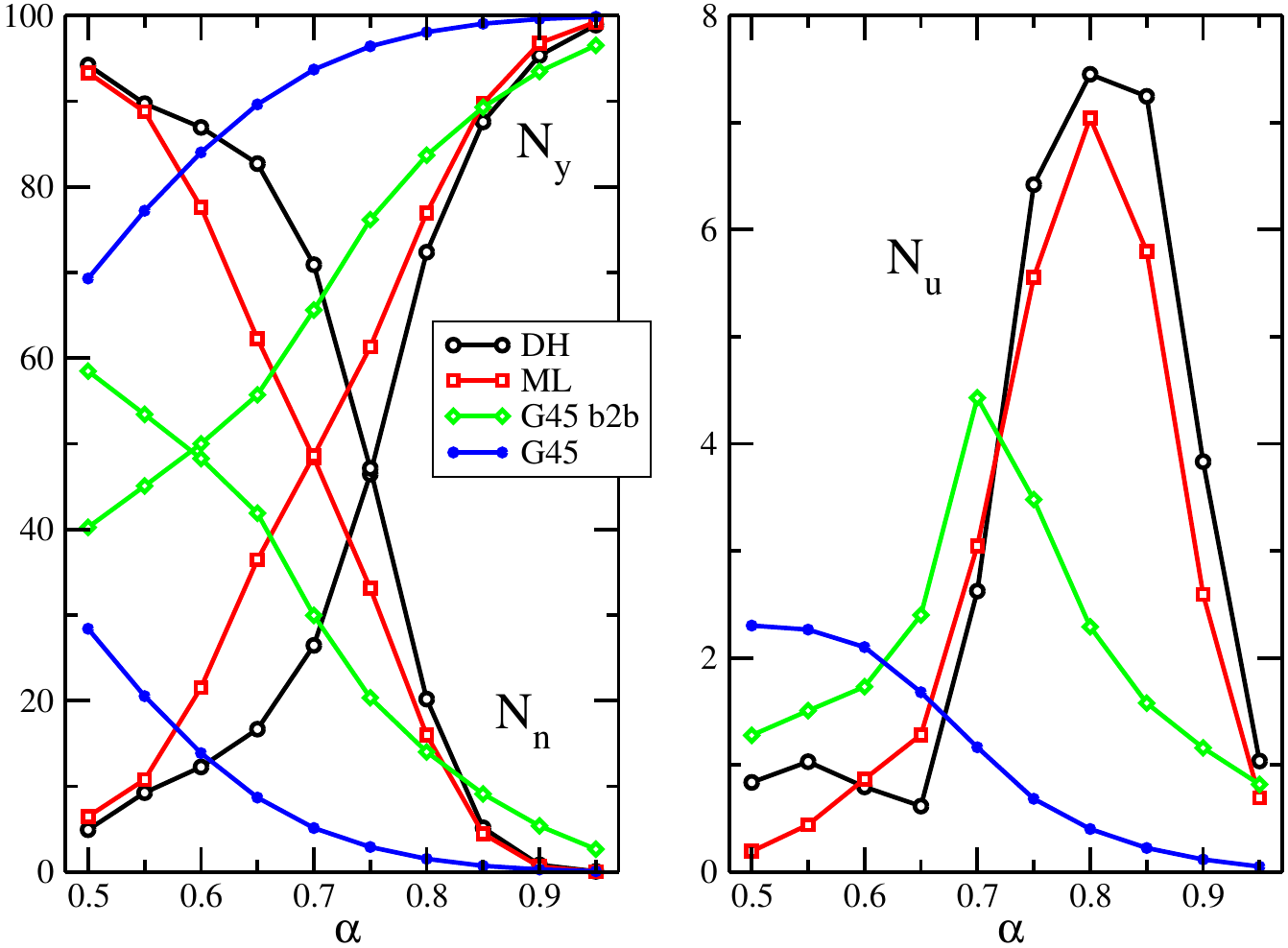}\\
Three weights:
\includegraphics[scale=0.36]{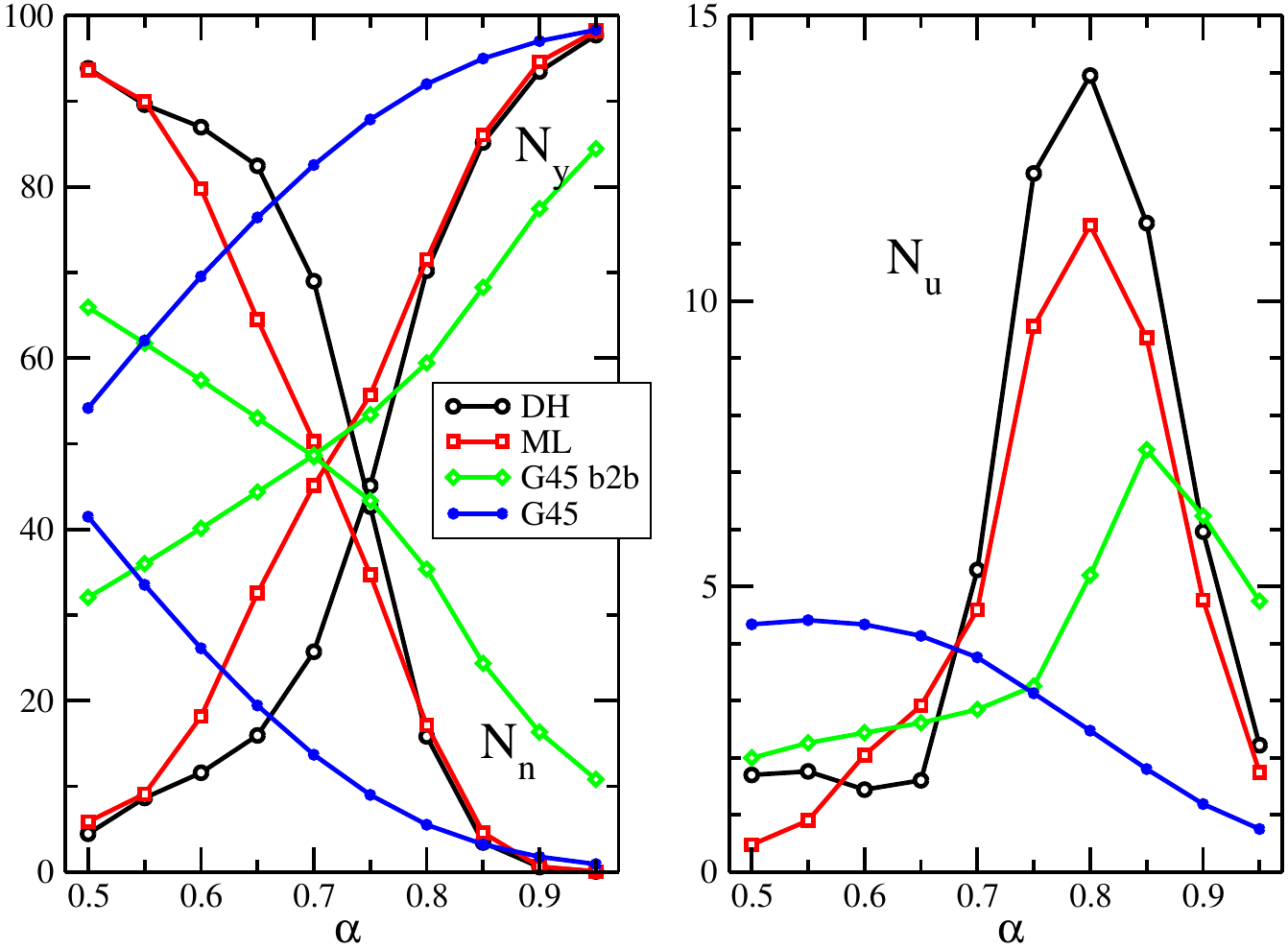}
\caption{"Single value of $p$ mode" routing. $N_y$, $N_n$, and $N_u$ as function of the coefficient $\alpha$
for different topologies: dual-home, mouth-like, and grid.
}
\label{r2}
\end{figure}

\begin{figure}[tp]
\centering
\begin{tabular}{cc}
Two weights: \hspace{1cm} & \hspace{1cm} Three weights: \\
\end{tabular}
\includegraphics[scale=0.33]{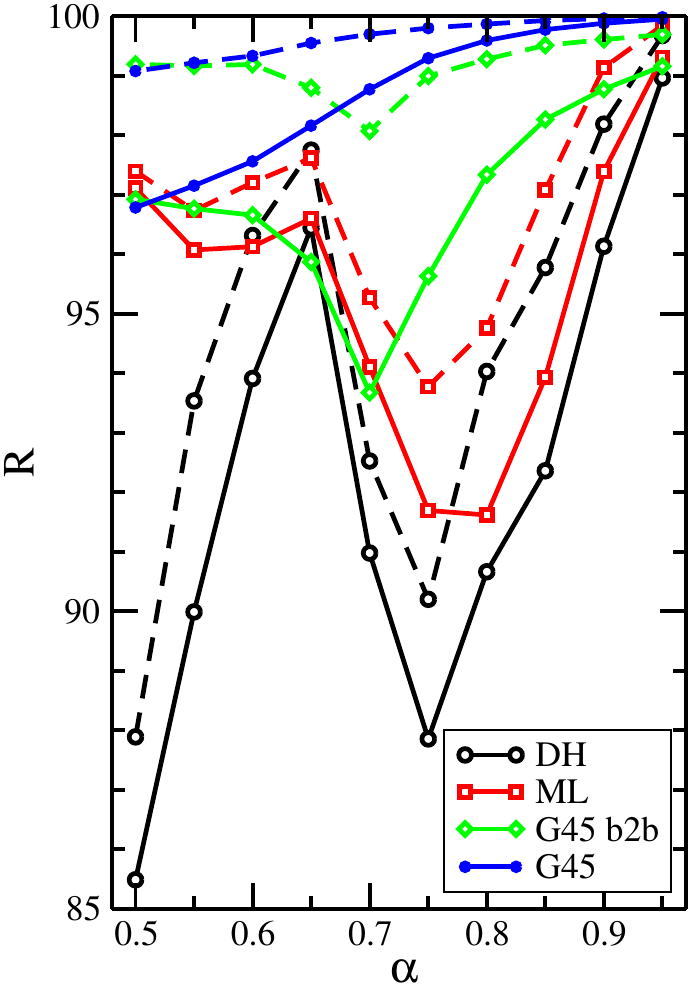}
\includegraphics[scale=0.33]{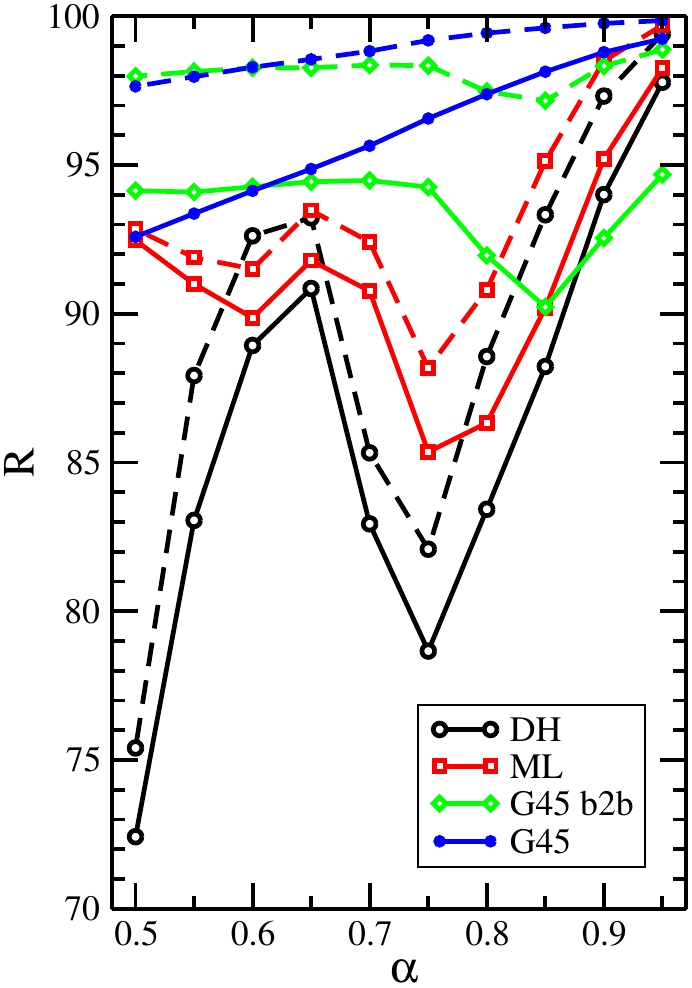}
\caption{Discovery rate $R(\alpha)$ for "single $p$ mode" (solid lines) and "multiple $p$ mode" (dashed lines) routing
for different topologies: dual-home, mouth-like, and grid.}

\label{r3}
\end{figure}

\section{Conclusion and discussion}

The problem of using a linear combination of weights for a distributed routing under multiple constraints was studied.
We showed that no $100\%$ feasible path discovery could generally be achieved using a distributed routing for a network,
even if a feasible path exists for every pair of the nodes.
We gave quantitatively sharp results in terms of the proportion of discovered paths and the slack in the constraints that is exhibited by feasible paths.  Moreover, we showed that in some cases no guarantee on discovery rate is possible in distributed mode.

We proposed an iterative approach to find an optimal value of parameter $p$ for a distributed routing.
We showed that for two randomly assigned weights in our scenario
the percent of uncertain pairs is $<$~5.2~\% in the "multiple values of $p$ mode"
cases and $<$~7.5~\% in the "single value of $p$ mode" cases for the dual-home and
mouth-like, and for the 45x45 grids is $<$~0.7~\% and $<$~2.3~\%, respectively.
This corresponds to $>$~99~\% "multiple $p$" and $>$~96~\% "single $p$" mode discovery
rate for 45x45 grids.
Theory suggests that our approach is nearly topology-independent and
can be generalized for any number of constraints. For example, the results for both 2-constraints and 3-constraints cases
exhibit similar behavior (Fig.~\ref{r1}--\ref{r3}). However, in the case of 3 constraints, the percent of uncertain, or undiscovered, pairs is about  twice higher than that for the case with 2 constraints.

We established concavity of single linear mixed metric as a function of parameter $p$ for a general case of $n$ constraints.
Obtained results generalized previous $2$-constraints case studied in \cite{Puri} and can be used while studying routing under
any number of constraints. Concavity (Theorem~\ref{thmconcave}) allows us to search for the optimal $p$ relatively fast,
in 5 to 10 iterations (Fig.~\ref{conv}).

\section{Ethical issues}
This work does not raise any ethical issues.
It does not involve human subjects or potentially sensitive data.


\bibliographystyle{IEEEtran}
\bibliography{reference}
\end{document}